\documentclass{article}
\usepackage[utf8]{inputenc}
\usepackage[english]{babel}

\usepackage{arxiv}
\usepackage{amsmath}
\usepackage{amsthm}

\usepackage[utf8]{inputenc} 
\usepackage[T1]{fontenc}    
\usepackage{hyperref}       
\usepackage{url}            
\usepackage{booktabs}       
\usepackage{amsfonts}       
\usepackage{nicefrac}       
\usepackage{microtype}      

\usepackage{graphicx}

\newtheorem{definition}{Definition}
\newtheorem{theorem}{Theorem}
\newtheorem{corollary}{Corollary}
\newtheorem{lemma}{Lemma}
\def\Mk{\mathcal M}

\title{A hierarchical structure in the motion representation of 2-state number-conserving cellular automata}

\author{
  KONG Gil-Tak\\
  Department of Information Engineering\\
  Hiroshima University\\
  Hiroshima, 739-0046\\
    \texttt{ggt3080@gmail.com} \\
   \And
    IMAI Katsunobu \\
 Department of Information Engineering\\
  Hiroshima University\\
  Hiroshima, 739-0046\\
    \texttt{imai@hiroshima-u.ac.jp} \\
  \AND
  NAKANISHI Toru\\
  Department of Information Engineering\\
  Hiroshima University\\
  Hiroshima, 739-0046\\
    \texttt{t-nakanishi@hiroshima-u.ac.jp} \\
}

\begin{document}
\maketitle

\begin{abstract}
A one-dimensional two-state number-conserving cellular automaton (NCCA) is a cellular automaton whose states are $0$ or $1$ and where cells take states 0 and 1 and updated their states by the rule which keeps overall sum of states constant. It can be regarded as a kind of particle based modeling of physical systems and has another intuitive representation, motion representation, based on the movement of each particle. We introduced a kind of hierarchical interpretation of motion representations to understand the necessary pattern size to each motion. We show any NCCA of its neighborhood size $n$ can be hierarchically represented by NCCAs of their neighborhood size from $n-1$ to $1$. 
\end{abstract}

\keywords{Cellular automata \and number conservation \and motion representation}

\section{Introduction}
Cellular automata(CA), introduced by von Neumann for modelling biological self-reproduction, is a discrete dynamical system which evolves in discrete space and time~\cite{Von}.
Among many kind of CA, a number-conserving CA (NCCA) has a feature that total number of states in a configuration is identical with that in the configuration of the previous time step~\cite{Boccara1998,Hattori}.
By the feature, NCCA can be a model to analyze a physical phenomenon with the property of mass conservation such as a traffic flow~\cite{Nagel}.

An NCCA can also be considered as a system of particle movements, i.e., each number in a cell represents the number of particles in the cell and the particles move to another cell simultaneously and each particle is not divided or disappeared. To describe the motions of particles, Boccara et al. proposed motion representation~\cite{Boccara1998,Boccara2002}. 
Although an NCCA has many variations of motion representations, Moreira et al.~\cite{Moreira04} introduced a canonical motion representation which is uniquely determined for an NCCA. 

An NCCA and a motion representation are inherently different computing models. For an NCCA, its neighborhood size is essential in contrast to a motion representation. Even the case of a two-state simple shift NCCA for a large neighborhood size $n$, the values for its rule table should be assigned for all $01$-patterns of length $n$ to determine the next state. But in the case of motion representation, just an information of a cell of state 1 is enough to identify the value 1 to be moved. Thus the only motion representation is enough to describe the simple shift CA for any neighborhood size.
The simplest car traffic rule 184~\cite{Nagel}, can be regarded as the combination of a basic shift and a motion depending on a size-two pattern, even the evolution can be embedded into an NCCA of any neighborhood size which is larger than three. 

Thus any motion representation of a two-state NCCA seems to be constructed by the set of motions which are ordered by their pattern size.
From this idea, we introduced a kind of hierarchical interpretation to motion representations to understand the necessary pattern size to each motion. 

In this paper, we show any NCCA of its neighborhood size n can be hierarchically represented by NCCAs of their neighborhood size from $n-1$ to $1$.

\section{Number-Conserving Cellular Automata}
A CA is a machine which cell's states are evolved through the interaction of each cells in a fixed area. The fixed area and the interaction are called neighborhood and rule. There are many kinds of CA according to neighborhood size, state of cells, etc. In this paper, we will deal with only 1-dimensional 2-state CA and we simply call it CA. 

\begin{definition}
[1-dimensional 2-state Cellular Automata]
A 1-dimensional 2-state cellular automaton $A$ is defined by $A = (n, f)$, where its neighborhood size $n$ is a non-negative finite integer and $f : \{0,1\}^{n} \rightarrow \{0,1\}$ is a mapping called the local function.
A configuration over $\{0,1\}$ is a mapping $c : \mathbf{Z} \rightarrow \{0,1\}$  where $\mathbf{Z}$ is the set of all integers. Then ${\rm Conf}(\{0,1\}) = \left\{ c | c : \mathbf{Z} \rightarrow \{0,1\}\right\}$ is the set of all configurations over $\{0,1\}$. The global function $F$ of $A$ is defined as $F : {\rm Conf}(\{0,1\}) \rightarrow {\rm Conf}(\{0,1\})$, i.e., 
\begin{equation*}
    \forall c \in {\rm Conf}(\{0,1\}),~\forall i \in \mathbf{Z} : F(c)(i) = f(c(i) \cdots c(i+n-1)).
\end{equation*}
\end{definition}
As above formula, the focus cell of CA is the left most cell in the neighborhood in this paper.
Note that we use the Wolfram coding \cite{Wolfram} $W(f)$ to represent a local function $f$: $W(f)=\sum
f(a_1\cdots a_n) 2^{2^{n-1} a_1 +2^{n-2} a_2+\cdots +2^0
  a_{n} }$ where the sum is applied on $\forall a_i \in 2^n (1\leq i\leq n)$. To represent a CA, we also use a pair of its neighborhood size and its Wolfram number instead of its local function. The local function $f$ is also referred to as the  \textit{rule} of $A$.
Fig.~\ref{fig:rule-ex} shows the rule table of CA (3, 226) as an example. 
\begin{figure}[h]
\begin{center}
\includegraphics[scale=0.4]{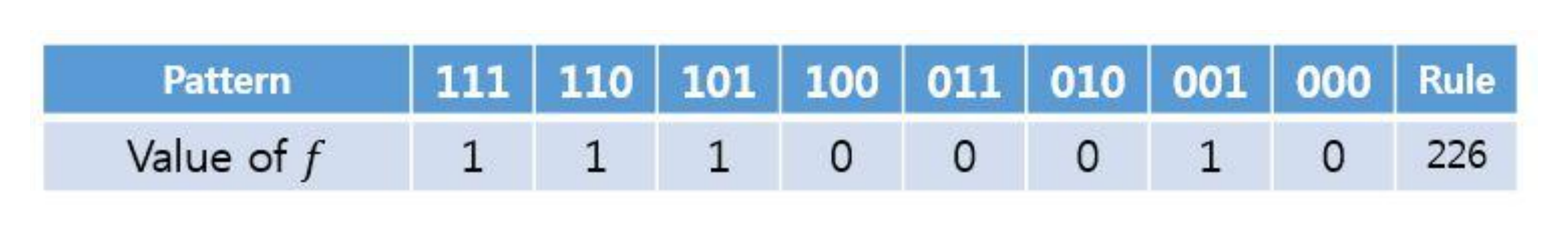}
\caption{The rule table of a 3-cell CA}
\label{fig:rule-ex}
\end{center}
\end{figure}

\begin{definition}
[Number Conserving CA]
A CA $A$ is said to be finite-number-conserving (FNC) iff
\begin{equation*}
    \forall c \in {\rm Conf}(\{0,1\}), \sum_{i\in \mathbf{Z}}\left\{F(c)(i)-c(i)\right\}=0
\end{equation*}
\end{definition}
In this paper, we only think about the case of finite configuration, i.e., the number of nonzero cells are finite. Because Durand et al.~\cite{Durand} showed that FNC is equivalent to the general infinite case, FNC is enough to show the number-conservation of a CA even for the case of infinite number of nonzero cells.

\begin{definition}[Pattern]
A pattern $p=a_1a_2\cdots a_k$ is a sequence of $a_i\in\{0, 1\}$ of a finite length $k$.  
In addition, for a pattern $p$, \ $\_\cdots\_p$ or $p\_\cdots\_$ or $\_\cdots \_p\_ \cdots \_$ is an extended pattern of $p$ where ``$\_\cdots\_$'' represents a finite sequence of the wildcard character ``\_'' which represents both $0$ and $1$.
\end{definition}
For example, $\_p_1p_2$ represents $0p_1p_2$ and  $1p_1p_2$ for any patterns $p_1$ and $p_2$.
Also we use the notation of concatenation of two or more patterns to represent another pattern. For example, if $p=010$, then $0p=0010$ and $p1=0101$.

In the following section, we denote a configuration $c=\cdots,c(i),c(i+1),\cdots\\,c(i+n-1),\cdots$ by $\cdots c_ic_{i+1}\cdots c_{i+n-1}\cdots$ as an abbreviation. 
We regard a sub configuration of finite size as a pattern and use it as the argument list of a local function $f$, i.e., we also denote $f(c(i), \cdots, c(i+n-1))$ by $f(c_i\cdots c_{i+n-1})$. 

We also use the notation, $|~~~|$, to represent the number of elements. For a pattern set $P$, $|P|$ means the number of patterns of $P$.
For a pattern $p$, $|p|$ means 
the length of $p$ as a pattern string.
But it is used in a slightly different way for configurations, i.e., for a configuration $c$, $|c|$ means the number of 1s.

%
\begin{definition}[Bundle] 
For a length $n-1$ pattern $r$, if length $n$ patterns $p$ and $q$ satisfy the condition:
\begin{equation*}
    p=0r,~q=1r \ ({\rm resp.} p=r0,~q=r1)
\end{equation*}
then we call $p(q)$  l-bundle (resp. r-bundle) of $r$. 
When $p,q,r$ satisfy the both cases, we call $r$ the bundle pattern of $p$ and $q$. 
\end{definition}

Next we define motion representation. 
Because an evolution of NCCA can be regarded as the movement of particles, there is another representation of an NCCA rule, motion representation \cite{Boccara1998,Moreira04}.

\begin{definition}
[Motion Representation]
Let $\tilde{p}$ be an extended pattern for invoking a motion $\mu$. A motion $\mu$ is defined as $\mu=(\tilde{p},s,e,v)$ where $s$ and $e$ is a start location and an end location, respectively. $v$ is a finite nonzero integer which represents a moving value from $s$ to $e$.
Let $\Mk$ be a set of motions $\{\mu_1,\dots,\mu_n\}$. For any configuration $c$ of an NCCA $A=(n, f)$ with the global function $F$, for each position in $c$ to which a translated $p_i$ (of $\mu_i$) matches, subtract $v_i$ from the cell $s_i$ and add $v_i$ to the cell $e_i$ simultaneously. If the resulting configuration is equal to $F(c)$ for any $c$, $\Mk$ is a motion representation of $A$.
\end{definition}

We also graphically represent a motion $\mu$ by an arrow over $\tilde{p}$ from $s$ to $e$ whose suffix is $v$. The suffix is omitted when $v$=1. It is shown that the motion representation of any 2-state CA can be composed by motions with $v=1$~\cite{ishizaka}. We briefly show the idea of the proof: suppose there is a motion $\mu_1=(p_1,s_1,e_1,2)$ in a motion representation. The start cell should be the destination of another motion, say $\mu_2=(p_2,s_2,e_2,v_2)$. Then if $v_2=1$ these motions can be replaced by $\mu'_1=(p_1,s_1,e_1,1)$ and $\mu_3=(p_3,s_3,e_3,1)$ where $p_3$ is a union of $\mu_1$ and $\mu_2$ and $s_3$ ($e_3$) is the related position to $s_2$ ($e_1$), respectively. If $v_2=2$ then $\mu'_2=(p_2,s_2,e_2,1)$ is also remained. In turn it is possible to replace all motions of $2$ (or more particles). Fig.~\ref{fig:motion} is an example of 3-cell motion representations.

\begin{figure}[h]
\begin{center}
\includegraphics[scale=0.24]{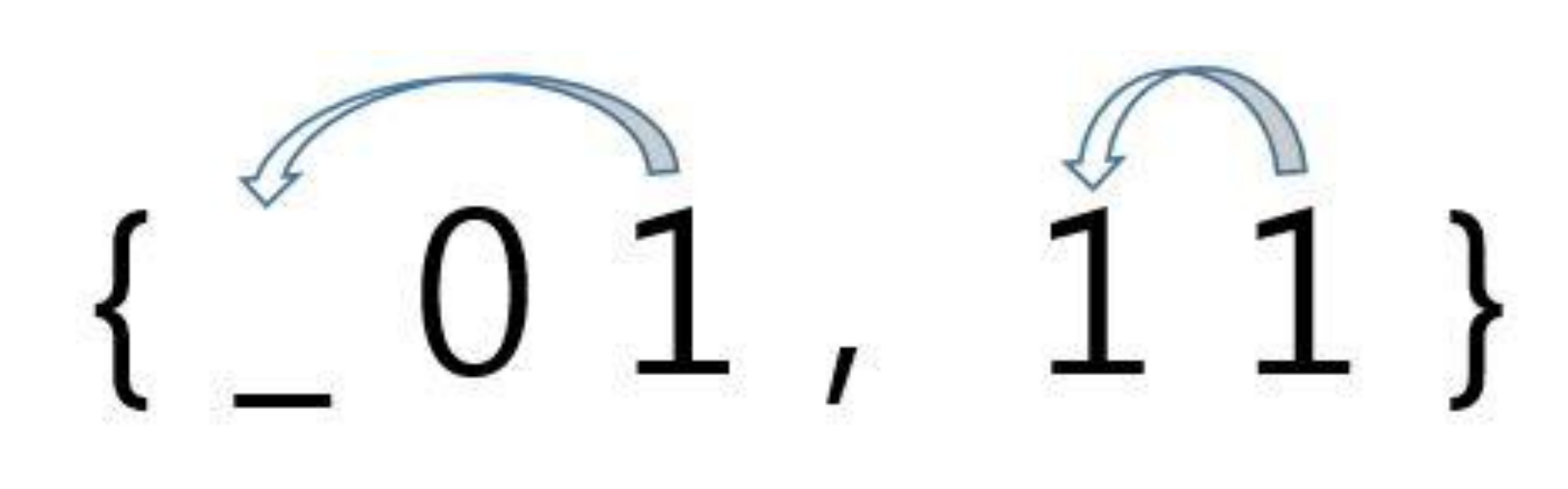}
\caption{Graphical representation of motion representation}
\label{fig:motion}
\end{center}
\end{figure}
\begin{figure}[h]
\begin{center}
\includegraphics[scale=0.6]{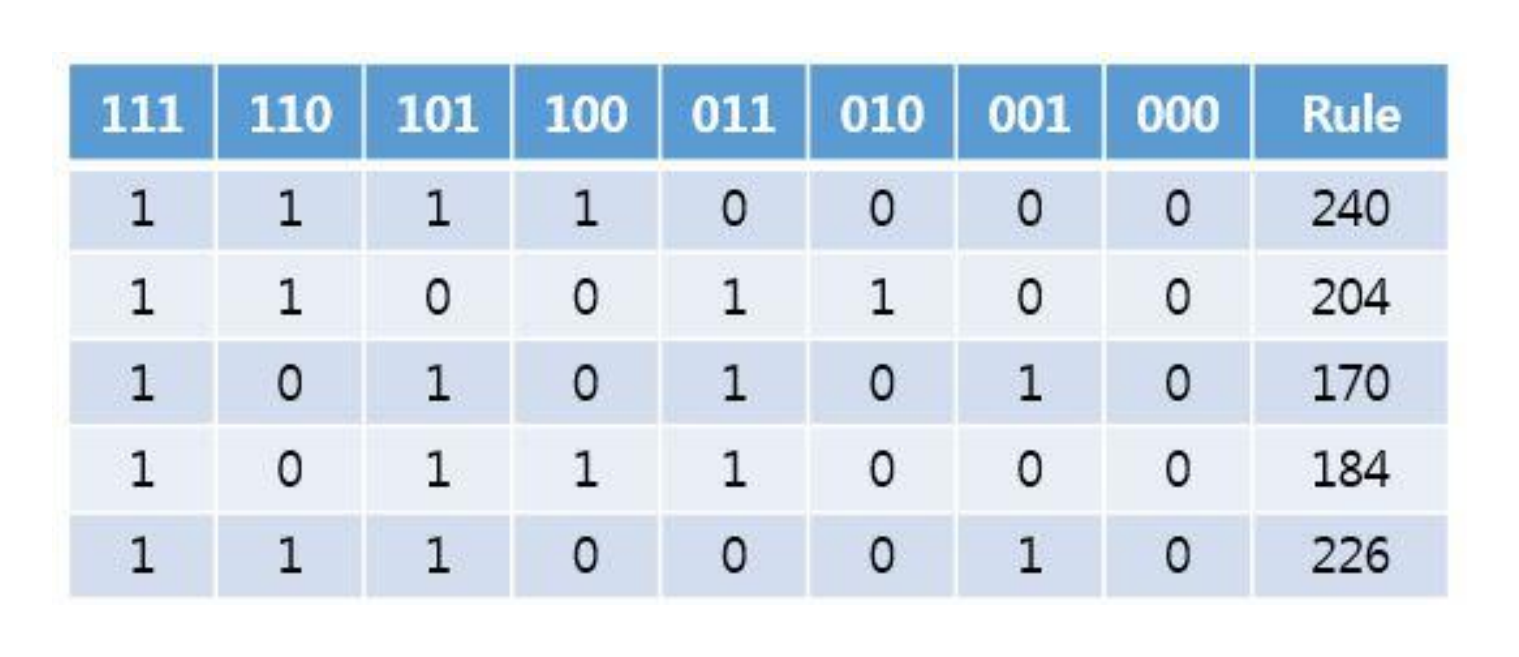}
\caption{3-cell NCCA rules}
\label{fig:3-cell-rules}
\end{center}
\end{figure}
\section{The Bundle tree}
In this section, we show a property of NCCA and the main principle of bundle tree for any NCCA.
\begin{definition}
[Value-1 pattern set]
For a CA $A=(n,f)$, 
we call the pattern set $P_A = \left\{p|f(p)=1\right\}$ the value-1 pattern set of $A$.
\end{definition}
\begin{definition}
[Bundle pattern set]
Let $P$ is a pattern set of length $n$ patterns.
For any $p=a_1\cdots a_n\in P$, if $q=\bar{a_1}a_2\cdots a_n$ (or $a_1\cdots a_{n-1}\bar{a_n}$) is in $P$ then
\begin{equation*}
        \hat{P} = \left\{ r | r=a_2\cdots a_n (~or~ a_1\cdots a_{n-1})\right\}
\end{equation*}
is a bundle pattern set of $P$ and $r$ is a bundle pattern of $p$ and $q$. If there is not such a pattern $q$ for a pattern in $P$ then there is no bundle pattern set of $P$. Also if $r(\in \hat{P})$ is a bundle pattern of $p$ and $q$ then $p$ and $q$ can not be l- or r-bundle of another pattern in $\hat{P}$.
\end{definition}
For example, $\left\{01,11\right\}$ can be a bundle pattern set of $\left\{001, 101, 110, 111\right\}$. Because $01$ can be a bundle pattern of $001, 101$ and $11$ can be a bundle pattern of $110, 111$. Moreover the following sets have no bundle pattern sets.\\
\begin{equation*}
    \left\{001, 010, 110, 111\right\}, \left\{011, 110, 111\right\}, \left\{01, 11, 110, 111\right\}
\end{equation*}

\begin{lemma}\label{lemma:bundle}
For a pattern $p$ in a value-1 pattern set of an NCCA, there exists a pattern $q$ in the set and a pattern $r$ s.t. $p$ and $q$ are either l-bundle or r-bundle of $r$.
\end{lemma}

\begin{proof}
Let $F$ be the global function of an NCCA $A=(n,f)$. 
Then $|F(\cdots0 c_1c_2 \cdots c_n 0\cdots)|=f(0\cdots 0 c_1)+f(0\cdots 0 c_1c_2)+\cdots +f(c_1c_2\cdots c_n)+\cdots +f(c_{n-1} c_n0\cdots 0) +f(c_n 0 \cdots 0)$.\\
Considering the configuration $c= \cdots 0 c_1 c_2 \cdots c_n \underbrace{0\cdots 0}_{k} \bar{c_1} c_2 c_3 \cdots c_n \underbrace{0\cdots 0}_{k} c_1 c_2 \cdots c_{n-1} \bar{c_n} 0 \cdots$($k>n$),  the following equation holds:
\begin{equation*}
   |c| = 2 \sum_{k=1}^n{c_k} + (\bar{c_1}+c_2+\cdots+c_{n-1}+\bar{c_n}),
\end{equation*}\\
where $\bar{c}$ is the negation of $c$ i.e., $c+\bar{c}=1$.
Since $f(0\cdots 0)=0$ then $|F(c)| = $
$|F(\cdots 0c_1\cdots c_n0\cdots)|+|F(\cdots 0\bar{c_1}c_2\cdots c_n0\cdots)|+|F(\cdots 0 c_1\cdots c_{n-1}\bar{c_n}0\cdots)|$

$=2|F(0\cdots 0c_1\cdots c_n0\cdots 0)|-f(c_1c_2\cdots c_n)+f(0\cdots c_1)+\cdots+f(\bar{c_1}c_2\cdots c_n)+f(c_1\cdots c_{n-1}\bar{c_n})+f(c_2\cdots c_{n-1}\bar{c_n}0)+\cdots+f(\bar{c_n}0\cdots 0)$.\\
Moreover the next formula holds because $F$ is the global function of an NCCA. 
\begin{equation}
    2 \sum_{k=1}^n{c_k} + (\bar{c_1}+c_2+\cdots+c_{n-1}+\bar{c_n}) = |F(c)|
\end{equation}
Because $2 \sum_{k=1}^n{c_k} =|F(0\cdots 0c_1\cdots c_n0\cdots 0)|$ and $\bar{c_1}+c_2+\cdots+c_{n-1}+\bar{c_n}=|F(0\cdots 0\bar{c_1}c_2\cdots c_{n-1}\bar{c_n}0\cdots 0)|$, we can get the following formula from (1):
\begin{equation}
    f(\bar{c_1}c_2\cdots c_{n-1}\bar{c_n})+f(c_1c_2\cdots c_n)
    = f(\bar{c_1}c_2\cdots c_n)+f(c_1\cdots c_{n-1}\bar{c_n})
\end{equation}
From (2), we can get the following result.\\
If $f(c_1c_2\cdots c_n)=1$https://www.overleaf.com/project/5da71725c5ddb5000130c4d7
then $f(\bar{c_1}c_2\cdots c_n)=1$ or $f(c_1c_2\cdots c_{n-1}\bar{c_n})=1.$
\end{proof}

By Lemma~\ref{lemma:bundle}, we can know that there are always bundle pattern for all patterns in value-1 pattern set of an NCCA.

When a pattern can be paired with two different elements in the value-1 pattern set $P$ of an NCCA, for example, $\{101, 100, 001\}\subset P$, we show that $000$ is also in $P$ by the next Lemma~\ref{half_lemma2}.
%
\begin{lemma} \label{half_lemma2}
Let $P_A$ is a value-1 pattern set of an NCCA $A(n,f)$. If three patterns $a_1\cdots a_n, a_1\cdots a_{n-1}\bar{a}_n, \bar{a}_1a_2\cdots a_n$ are in $P_A$ then pattern $\bar{a}_1a_2\cdots a_{n-1}\bar{a}_n$ is also in $P_A$.
\end{lemma}
\begin{proof}
Suppose $\bar{a}_1a_2\cdots a_{n-1}\bar{a}_n$ is not in $P_A$ with $a_1\cdots a_n, a_1\cdots a_{n-1}\bar{a}_n, \bar{a}_1a_2\cdots a_n \in P_A$. In other words, $f(a_1\cdots a_n)= f(a_1\cdots a_{n-1} \bar{a}_n)=f(\bar{a}_1a_2\cdots a_n)=1$ and $f(\bar{a}_1a_2\cdots a_{n-1}\bar{a}_n)=0$.\\
According to NCCA principles, we can get following formulas;\\
$a_1=f(a_1\cdots a_n)+f(0a_1\cdots a_{n-1})+\cdots +f(0\cdots 0a_1)
-\left\{f(0a_2\cdots a_n)+\cdots +f(0\cdots 0 a_2)\right\}\\
~~~~=f(a_1\cdots a_{n-1}\bar{a}_n)+f(0a_1\cdots a_{n-1})+\cdots +f(0\cdots 0a_1)
-\left\{f(0a_2\cdots a_{n-1}\bar{a}_n)+\cdots +f(0\cdots 0a_2)\right\}$\\
Then $0=f(a_1\cdots a_n)-f(0a_2\cdots a_n)-\{f(a_1\cdots a_{n-1}\bar{a}_n)-f(0a_2\cdots a_{n-1}\bar{a}_n\}$.\\
Since $0=f(a_1\cdots a_n)=f(a_1\cdots a_{n-1}\bar{a}_n)=1$,
\begin{equation}
    f(0a_2\cdots a_n)=f(0a_2\cdots a_{n-1}\bar{a}_n)
\end{equation}
On the same way,\\
$\bar{a}_1=
f(\bar{a}_1a_2\cdots a_n)+f(0\bar{a}_1a_2\cdots a_{n-1})+\cdots +f(0\cdots 0\bar{a}_1)-\{f(0a_2\cdots a_n)+\cdots +f(0\cdots 0a_2)\}$\\
$~~~~~=f(\bar{a}_1a_2\cdots a_{n-1}\bar{a}_n)+f(0\bar{a}_1a_2\cdots a_{n-1})+\cdots +f(0\cdots 0\bar{a}_1)-\{f(0a_2\cdots a_{n-1}\bar{a}_n)+\cdots +f(0\cdots 0a_2)\}$.\\

Then $0=f(\bar{a}_1a_2\cdots a_n)-f(0a_2\cdots a_n)-\{f(\bar{a}_1a_2\cdots a_{n-1}\bar{a}_n)-f(0a_2\cdots a_{n-1}\bar{a}_n\}$.\\
Since $f(\bar{a}_1a_2\cdots a_n)=1,~ f(\bar{a}_1a_2\cdots a_{n-1}\bar{a}_n)=0$,
\begin{equation}
   f(0a_2\cdots a_n)=1+f(0a_2\cdots a_{n-1}\bar{a}_n)
\end{equation}
Formula (3) and (4) be a contradiction.
\end{proof}

In the process of making a bundle pattern set of the value-1 pattern set of an NCCA, a pattern might be paired with two different patterns. For example, the value-1 pattern set of an NCCA rule 204 is $\{010, 011, 110, 111\}$. In this case, $010$ can not only make a pair with $011$ to $01$ but also $110$ to $10$. But by Lemma~\ref{half_lemma2}, there exist another pattern $111$ which can be paired with $011$ and $110$. Therefore we can make two distinct pairs like $\{010, 011\}, \{110, 111\}$ or $\{010, 110\}, \{011, 111\}$.
Thus, every patterns in value-1 pattern set can be paired without any overlapping patterns by Lemma~\ref{half_lemma2}.

For the value-1 pattern set of an NCCA $P$ and for a pattern $r$, suppose that $\{0r,1r,r0,r1\}\subset P$. Then the bundle pattern of $0r,1r$ and $r0,r1$ is both $r$. But in such case, there exists another element which can be paired with one of $0r,1r,r0,r1$ by the next Lemma. Thus all bundle patterns from a value-1 pattern set can be different.

\begin{lemma}
Let $P_A$ be a value-1 pattern set of an NCCA $A=(n,f)$. If
$0a_1\cdots a_{n-1}$, $1a_1\cdots a_{n-1}$, $a_1\cdots a_{n-1}0$, $a_1\cdots a_{n-1}1$ $\in P_A$ then at least one pattern among $0a_1\cdots a_{n-2}\bar{a}_{n-1}$,  $1a_1\cdots a_{n-2}\bar{a}_{n-1}$, $\bar{a}_1a_2\cdots a_{n-1}0$, $\bar{a}_1a_2\cdots a_{n-1}1$ be an element of $P_A$ where $a_i \in \{0,1\}$.
\end{lemma}
\begin{proof} Omitted. 
\end{proof} 

\medskip

Therefore for all NCCA $A$, there exist a bundle pattern set $\hat{P}_A$ of half the size of its value-1 pattern set $P_A$.

\begin{theorem}\label{thm:main}
For an $n$-cell NCCA $A$ ($n\ge 2$) with $|\hat{P}_A|=|P_A|/2$, an $(n-1)$-cell CA $B$ satisfying $P_B=\hat{P_A}$ is an $(n-1)$-cell NCCA.
\end{theorem}

\begin{proof}
Let $F$ be the global function of an NCCA $A=(n,f)$. Because $A$ is an NCCA, $|c|=|F(c)|$ for any configuration $c(=\cdots c_{-1}c_0c_1\cdots)$. 
Let $G$ be the global function of a CA $B=(n-1,g)$.
If $|G(c)|=|F(c)|$ then $|c|=|F(c)|=|G(c)|$. i.e., $B$ can be an NCCA.
Then we will show $|G(c)|=|F(c)|$ from now.\\
For each state-1 cell $F(c)(k)=1$, $k\in \mathbb{Z}$ on $F(c)$,
\begin{equation*}
    \exists p = c_k \cdots c_{k+n-1} \in P_A \ {\rm s.t.\ } f(p)=1.
\end{equation*}
Also for a pattern $q \in \hat{P}$, which is a bundle pattern of $p$, it can be satisfied $g(q)=1$(i.e., $q \in P_B$).
Thus if $F(c)(k)=1$ then either $G(c)(k)=1$ or $G(c)(k+1)=1$ holds.
When $p$ is a l-bundle of $q$, $G(c)(k)=1$ and when $p$ is a r-bundle of $q$, $G(c)(k+1)=1$ like Fig.~\ref{fig:bundle}.\\
\begin{figure}[h]
\begin{center}
\includegraphics[scale=0.22]{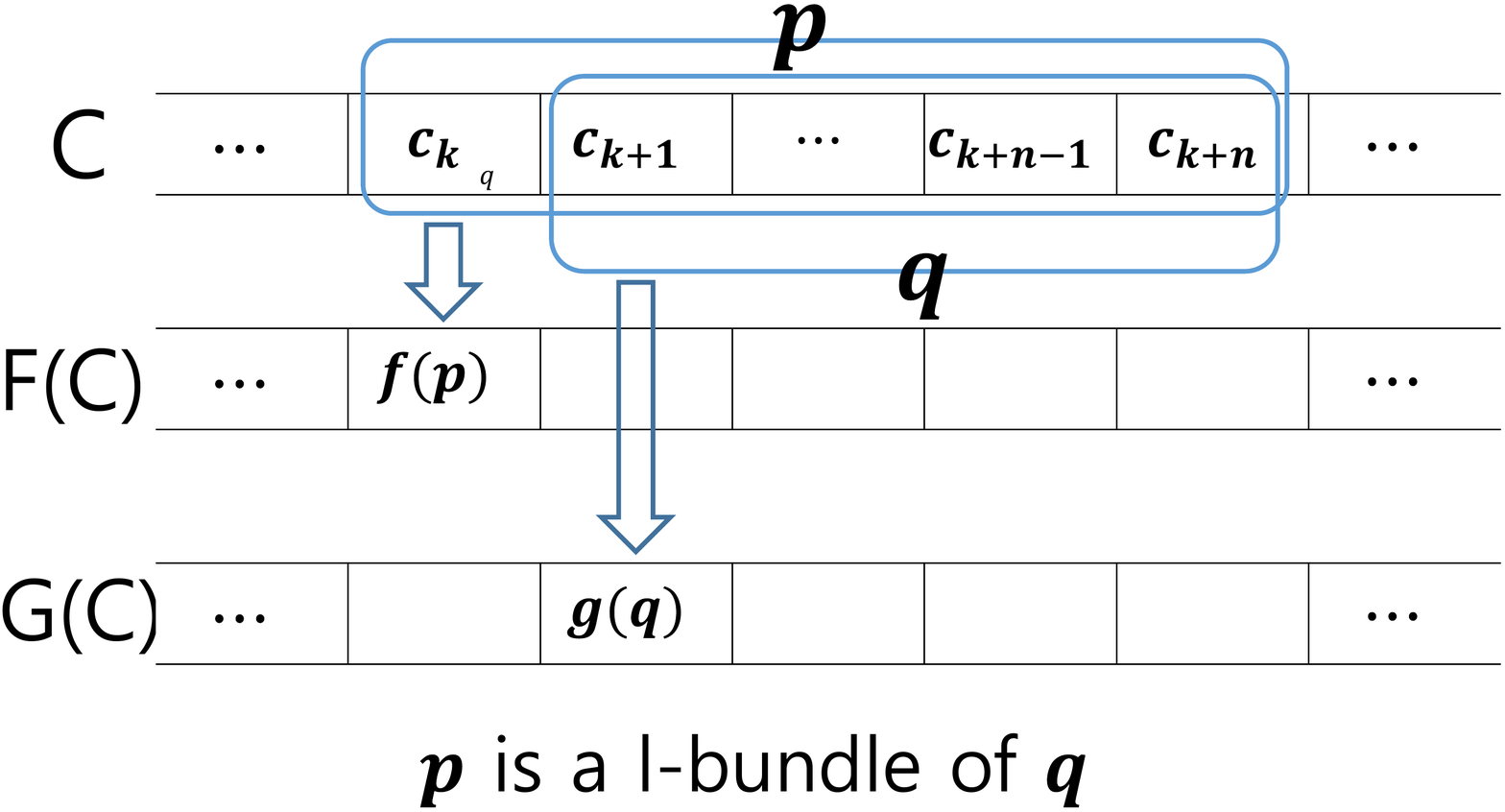}
\includegraphics[scale=0.22]{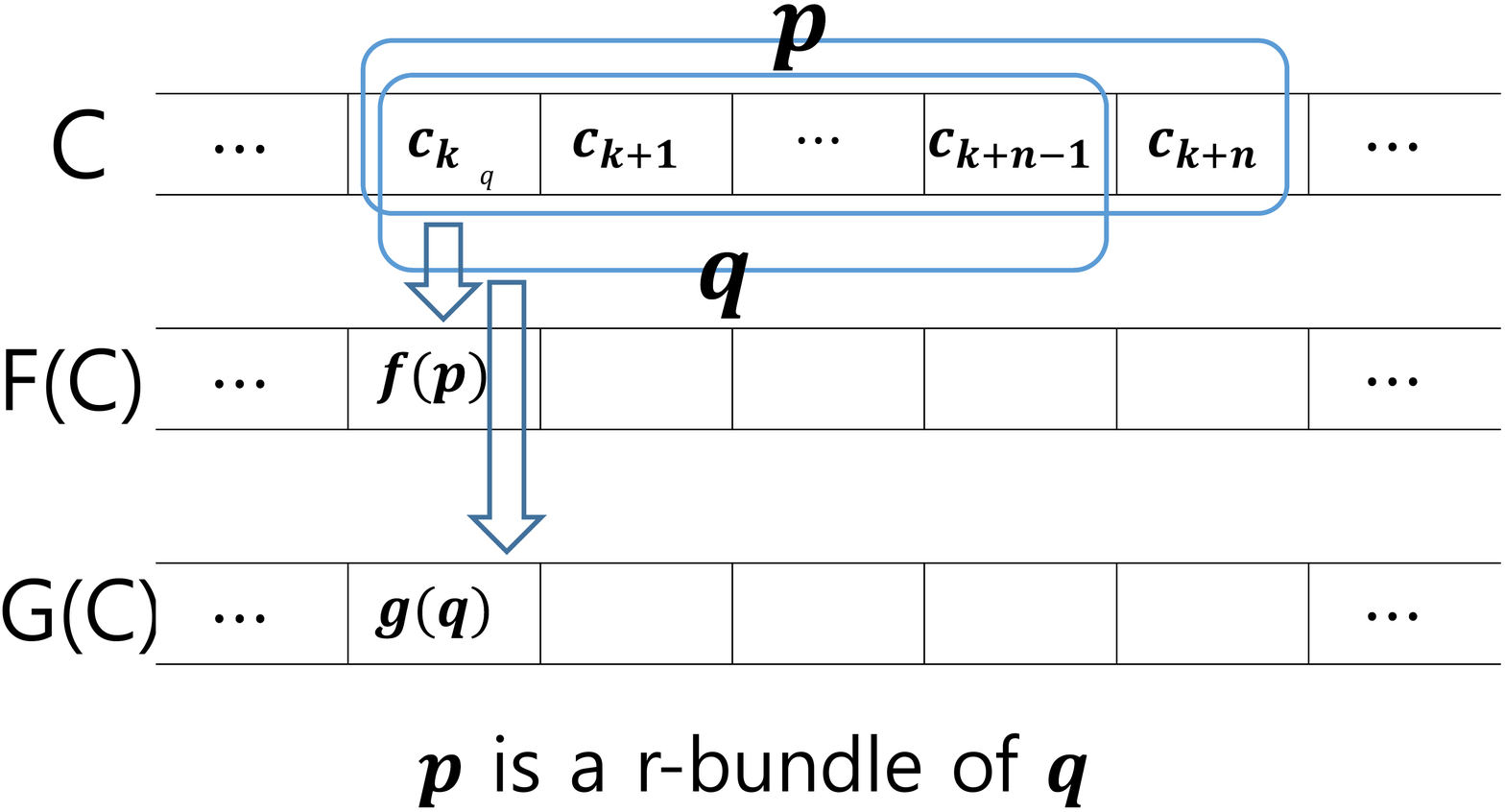}
\caption{Two evolutions according to the relation between $p$ and $q$}
\label{fig:bundle}
\end{center}
\end{figure}

\noindent If $F(c)(k)=F(c)(k+1)=1$ then\\
(1) by $F(c)(k)=1$, either $G(c)(k)=1$ or $G(c)(k+1)=1$ holds,\\
(2) by $F(c)(k+1)=1$, either $G(c)(k+1)=1$ or $G(c)(k+2)=1$ holds.\\
So if $G(c)(k+1)=1$ occur simultaneously in (1) and (2), an overlap occurs.
Thus $|F(c)|=|G(c)|$ is satisfied if there is no overlap.\\
Suppose that the above case has occurred. Then there are four patterns $p=c_k \cdots c_{k+n-1}$, $q=c_{k+1}\cdots c_{k+n-1}$ which is a l-bundle of $p$ and $p'=c_{k+1}\cdots c_{k+n}$, $q'=c_{k+1}\cdots c_{k+n-1}$ which is a r-bundle of $p$. Then $q$ is the same with $q'$. Because of two patterns $q,q'$ are the same in $P_B$, $|P_B|=|P_A|/2 - 1$. This contradics the assumption, $|P_B|=|P_A|/2$.\\
In the result, $G(c)$ is a 2-state configuration with $|G(c)|=|F(c)|=|c|$. Then $B$ be an NCCA. 

The value-1 pattern set of 1-cell NCCA is $\{1\}$, so the number of elements in $P_A$ of $n$-cell NCCA ($n\ge 2$) is always $2^{n-1}$.
\end{proof}

Theorem~\ref{thm:main} shows that the value-1 pattern sets of an $n$-cell and an $(n-1)$-cell NCCA have a kind of hierarchical relation and we can extract the relation as a tree structure as follows:

\noindent
Let $P_n$ be the value-1 pattern set of an NCCA $A_n=(n,f_n)$.
By Theorem~\ref{thm:main}, we can get a sequence of value-1 pattern sets $P_i (n\ge i\ge 1)$ of $A_i=(i,f_i)$ where 
$P_{i-1}=\hat{P}_i$ and $|P_i|=2^{i-1}$. 
For each element $r$ in $P_i~(1<i<n-1)$, there are two elements $p,q\in P_{i+1}$ where $p (q)$ is an l-(r-)bundle of $r$, respectively. We can construct a tree $T_{A_n}=(V,E)$ where $V$ is the set of all elements in all sets $P_i$ and $E$ is the set of all edges 
$(r,p)$ and $(r,q)$ described above. 
Clearly, $T_{A_n}$ is a complete binary tree and its root vertex is $1$ and its height is $n-1$. The height-$i$ vertices of  $T_{A_n}$ are the elements of the value-1 pattern set of $A_i$. We call $T_{A_n}$ a {\em bundle tree} of $A_n$.
Fig.~\ref{fig:bundletree} is a bundle tree of the CA $(4,62600)$.

\begin{figure}[h]
\begin{center}
\includegraphics[scale=0.32]{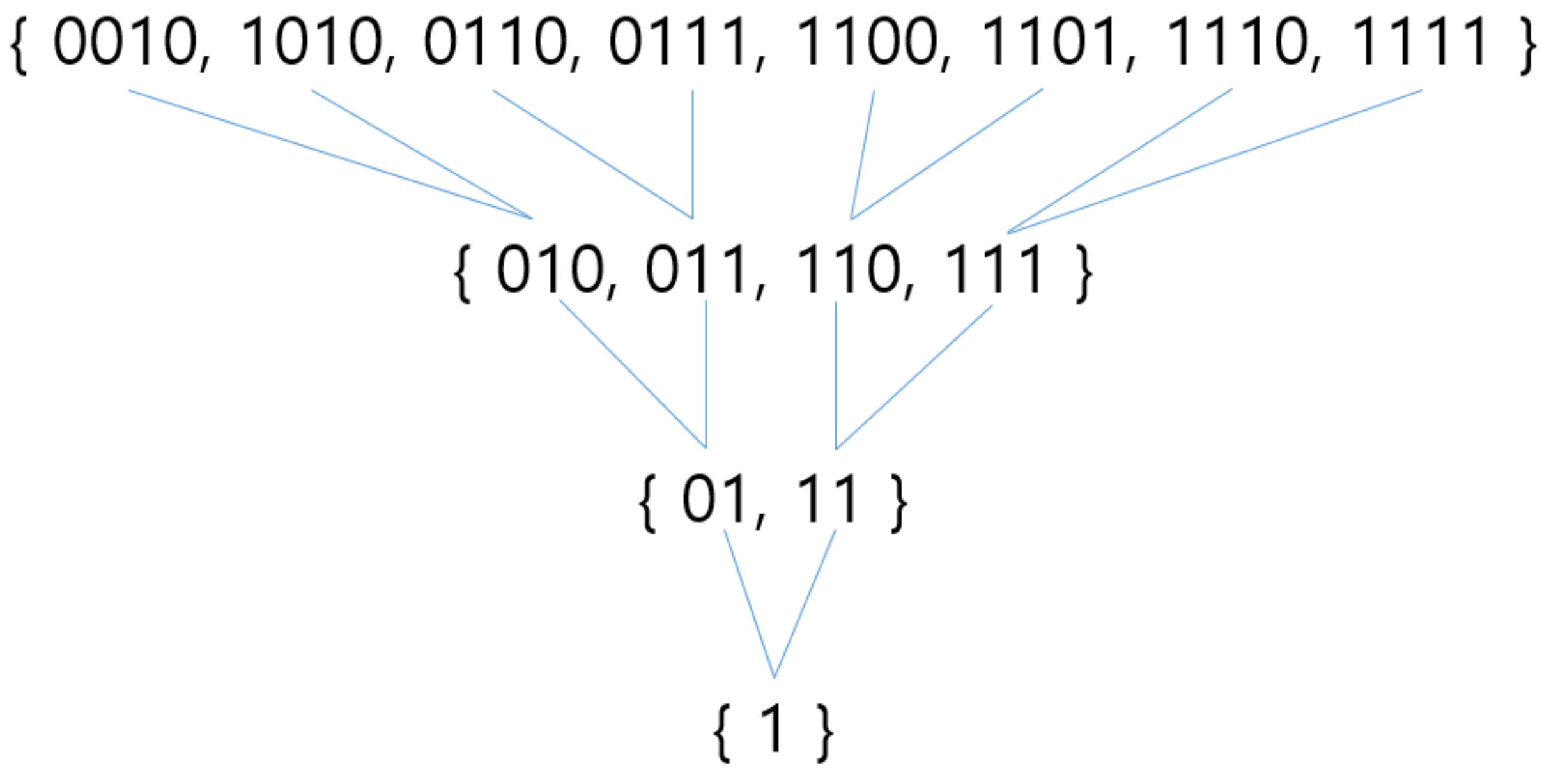}
\caption{A bundle tree of (4,62600)}
\label{fig:bundletree}
\end{center}
\end{figure}

By Lemma~\ref{lemma:bundle}, it is clear that a bundle tree exists for any NCCA rules.
Moreover we can get Corollary~\ref{cor:bundletree} and Theorem~\ref{theorem2} by Theorem~\ref{thm:main}.
\begin{corollary}\label{cor:bundletree}
For a bundle tree of $A_n$, $\forall i \in \mathbb{Z}, A_i$ is an $i$-cell NCCA.
\end{corollary}
\begin{theorem}
\label{theorem2}
Bundle tree of any NCCA is always a binary tree with root $\left\{1\right\}$.
\end{theorem}
\begin{proof}
The number of elements of value-1 pattern set of an NCCA $(n,f)$ is always $2^{n-1}$. Then it is clear that the number of elements of pattern sets at $i$th level is $2^{i-1}$ by corollary2. Moreover the smallest cell NCCA is 1-cell NCCA ${1}$. Then bundle tree of any NCCA be a binary tree with root(the $1$st level) ${1}$. 
\end{proof}

\section{Conclusion}
In this paper, we show any NCCA of its neighborhood size n can be hierarchically represented by NCCAs of their neighborhood size from $n-1$ to $1$.
The result supports that any NCCA can be understood by a hierarchical motion representations according to the necessary pattern size to each motion.

\bibliographystyle{unsrt}  


\end{document}